\documentclass[a4paper]{article}
\usepackage{amsfonts}
\usepackage{amsthm}
\usepackage[utf8]{inputenc}
\usepackage[english]{babel}
\usepackage{hyperref}

\newtheorem{theorem}{Theorem}[section]
\newtheorem{definition}{Definition}[section]
\newtheorem{example}{Example}[section]

\title{Fuzzy Maximum Satisfiability}

\author{Mohamed El Halaby \\ \small{Department of Mathematics}\\
\small{Faculty of Science}\\
\small{Cairo University}\\
\small{Giza, 12613, Egypt}\\
\href{mailto:halaby@sci.cu.edu.eg}{halaby@sci.cu.edu.eg}
\and
Areeg Abdalla \\ \small{Department of Mathematics}\\
\small{Faculty of Science}\\
\small{Cairo University}\\
\small{Giza, 12613, Egypt}\\
\href{mailto:areeg@sci.cu.edu.eg}{areeg@sci.cu.edu.eg}}

\date{}

\begin{document}
\maketitle
\begin{abstract}
In this paper, we extend the Maximum Satisfiability (MaxSAT) problem to \L{}ukasiewicz logic. The MaxSAT problem for a set of formulae $\Phi$ is the problem of finding an assignment to the variables in $\Phi$ that satisfies the maximum number of formulae. Three possible solutions (encodings) are proposed to the new problem: (1) Disjunctive Linear Relations (DLRs), (2) Mixed Integer Linear Programming (MILP) and (3) Weighted Constraint Satisfaction Problem (WCSP). Like its Boolean counterpart, the extended fuzzy MaxSAT will have numerous applications in optimization problems that involve vagueness.
\end{abstract}

\section{Introduction}
\textit{Boolean Satisfiability} (SAT) stands at the crossroads of logic, graph theory and computer science in general. For this reason, nowadays more problems are being solved faster by SAT solvers than other means. A lot of real-life problems are difficult to solve because they pose computational challenges. In many of these problems, it is not sufficient to find a solution but rather one that is optimal. These are called optimization problems and they arise frequently in the real world. One of the most effective ways to solve optimization problems is to first model them mathematically or logically, then solve them using a suitable algorithm.

Maximum Satisfiability (MaxSAT) is the optimization version of SAT. In recent years, there has been a growing interest in developing efficient  algorithms\cite{morgado2013iterative,narodytska2014maximum} and implementing them into competent solvers that could solve instances from real-life applications\cite{asin2012curriculum,jose2011cause,safarpour2007improved,argelich2010solving,argelich2008cnf}. In fact, an annual competition called The MaxSAT Evaluations is held for the purpose of running recent solvers on categories of benchmarks (random, crafted and industrial instances) then declaring a winner for each categories. Indeed, the performance of MaxSAT solvers is getting better with time and hence it is becoming more feasible to solve practical problems using MaxSAT.

Fuzzy logic is an extension of Boolean logic by Lotfi Zadeh in 1965 based on the theory of fuzzy sets, which is a generalization of the classical set theory. Introducing the notion of degree in the verification of a condition enables a condition to be in a state other than true or false (thus, infinite truth degrees). Fuzzy logic provides a very valuable flexibility for reasoning, which makes it possible to take into account inaccuracies and vagueness.

\subsection{Boolean Logic and SAT}

A \textit{Boolean variable} $x$ can take one of two possible values: $1$ or $0$. A \textit{literal} $l$ is a variable $x$ or its negation $\neg x$. A \textit{disjunction} $C$ is a group of $r$ literals joined by $\lor$. This is expressed as $C = \bigvee_{i=1}^{r} l_i$ A Boolean formula $\phi$ in Conjunctive Normal Form (CNF) is a group of $m$ disjunctions joined by $\wedge$ (i.e., a conjunction of disjunctions). From now on, we will refer to a disjunction in a CNF formula as a \textit{clause}. If $\phi$ consists of $m$ clauses where each clause $C_i$ is composed of $r_i$ literals, then $\phi$ can be written as
$$\phi = \bigwedge_{i=1}^m C_i$$ where
$$C_i = \bigvee_{j=1}^{r_i} l_{i,j}$$

A formula is said to be $k$-CNF if each clause has exactly $k$ literals. Sometimes we consider a CNF formula as a set of clauses $\phi=\{C_1,\dots,C_m\}$. A Boolean CNF formula will be referred to as just a formula for short. If $\phi$ is a formula over the $n$ variables $x_1,\dots,x_n$, then a \textit{complete assignment} of $\phi$ is a set $A=\{x_1=b_1,\dots,x_n=b_n\}$, where each $b_i,(1 \leq i \leq n)$ is either $1$ or $0$. A \textit{partial assignment} is an assignment that leaves out some variables unassigned. An assignment $A$ (complete or partial) satisfies a literal $x$ if $x$ is assigned $1$ in $A$ and satisfies a literal $\neg x$ if $x$ is assigned $0$ in $A$. A clause $C$ is satisfied by $A$ if at least one literal of $C$ is satisfied by $A$. A formula $\phi$ is satisfied by $A$ if $A$ satisfies all the clauses of $\phi$.

The decision version of the \textit{SAT problem}, given a formula $\phi$, is deciding whether there exists an assignment that satisfies $\phi$. The search version is concerned with finding (searching) for a satisfying assignment for $\phi$. For example, $\phi_1=\{(x_1 \vee \neg x_2),(\neg x_1 \vee x_3),(\neg x_1 \vee x_2 \vee \neg x_3)\}$ has the satisfying assignment $A_1=\{x_1=1,x_2=0,x_3=0\}$. The formula $\phi_2=\{(x_1 \lor x_2 \lor x_3 \lor x_4 \lor \neg x_5),(x_1 \lor x_2 \lor x_3 \lor \neg x_4 \lor x_5),(x_1 \lor x_2 \lor \neg x_3 \lor x_4 \lor x_5),(x_1 \lor \neg x_2 \lor x_3 \lor x_4 \lor x_5),(\neg x_1 \lor x_2 \lor x_3 \lor x_4 \lor x_5),(x_1 \lor x_2 \lor x_3 \lor x_4 \lor x_5)\}$ has the solution $A_2 = \{x_1 = 0, x_2 = 1, x_3 = 1, x_4 = 0, x_5 = 0\}$, which indeed satisfies $\phi_2$.

\subsection{MaxSAT}
\textit{Maximum Satisfiability} is a generalization of SAT. The idea behind it is that sometimes not all restrictions of a problem can be satisfied, and we try to satisfy as much of them as possible.

Given a CNF formula $\phi$, MaxSAT asks for an assignment that maximizes the number of satisfied clauses{pipatsrisawat2007clone}.  For example, $\phi=\{(y\vee z),(\neg z),(x\vee\neg y),(\neg x\vee z)\}$ has $A=\{x=0,y=0,z=0\}$ as a solution. The maximum number of satisfied clauses in $\phi$ is three. Table 1.1 shows all the possible assignments for $\phi$ and the number of clauses that each one satisfies. 

\begin{table}[h]
\centering
\begin{tabular}{|c|c|c|c|}

\hline
{$\textbf{x}$} & {$\textbf{y}$} & {$\textbf{z}$} & {\textbf{Number of satisfied clauses}} \\ \hline
{ $0$} & {$1$} & {{$1$}} & {{2}}       \\ \hline
{  $0$} & {  $1$} & {  $0$} & {3}       \\ \hline
{  $0$} & {  $0$} & {  $1$} & {3}       \\ \hline
{  $0$} & {  $0$} & {  $0$} & {3}       \\ \hline
{  $1$} & {  $1$} & {  $1$} & {3}       \\ \hline
{  $1$} & {  $1$} & {  $0$} & {3}       \\ \hline
{  $1$} & {  $0$} & {  $1$} & {3}       \\ \hline
{$1$} & {$0$} & {$0$} & {2}       \\ \hline
\end{tabular}
\caption{All possible assignments for $\phi$}
\label{table:table1}
\end{table}

There are two general techniques to solve MaxSAT: (1) branch and bound algorithms, and (2) SAT-based algorithms. Branch and bound algorithms\cite{borchers1998two} work by searching the binary tree of all partial assignments to the variables of the input formula. The procedure starts with the empty assignment at the root of the tree and traverses it in a depth-first manner to find an optimal complete solution (represented by leaf nodes). Branching occurs on an unassigned variable at some node and the children of this node correspond to assigning the variable 1 or 0. Later works added more effective techniques in order to boost the search. Namely, more efficient data-structures, new branching heuristics, new simplification rules and more accurate lower bounds\cite{niedermeier1999new,alsinet2004max,shen2005improving,xing2005maxsolver,li2007new,larrosa2008logical}. In practice, branch and bound Max-SAT solvers are suitable for instances generated at random and some crafted ones.

SAT-based MaxSAT algorithms\cite{fu2006solving,marques2007using,marques2008towards,le2010sat4j,an2011qmaxsat} are based on iteratively calling a SAT solver. These techniques work by maintaining and refining a lower bound and/or an upper bound to the optimal solution with the help of a SAT-solver. It has been found that these techniques are particularly suitable for benchmarks coming from industrial applications and some crafted ones.  One way to do this, given a Max-SAT instance, is to check if there is an assignment that falsifies no clauses. If such an assignment can not be found, the algorithm checks if there is an assignment that falsifies only one clause. This is repeated and each time the algorithm increments the number of clauses that are allowed to be falsified until the SAT solver returns 1 (or true), meaning that the minimum number of falsified clauses has been determined. Comprehensive surveys on SAT-based MaxSAT solving can be found in\cite{morgado2013iterative,ansotegui2013sat}.

\subsection{Fuzzy Logic}
Let $X$ be a nonempty set, a fuzzy set $A$ in $X$ is characterized by its membership function $$\mu_A: X \rightarrow [0,1]$$ and $\mu_A(x)$ is interpreted as the degree of membership of element $x$ in fuzzy set $A$ for each $x \in X$. So, $A$ is determined by $$A=\{(x,\mu_A(x))\mid x\in X \}$$

A fromula is built from a set of variables $\mathcal{V}$, constants from $[0,1]$ and an $n$-ary connective $F$ for $n \in \mathbb{N}$. An assignment (also called an \textit{interpretation})is a mapping $I: \mathcal{V} \rightarrow [0,1]$, where:
\begin{itemize}
	\item For each constant $c \in [0,1]$, $[c]_I=c$.
	\item $[\neg \phi]_I=F_\neg([\phi]_I)$.
	\item $[\phi \circ \psi]_I = F_\circ([\phi]_I \circ [\psi]_I)$, where $\circ \in \{\neg,\oplus,\odot,\wedge,\lor\}$ is a binary connective.
\end{itemize}

The following table defines basic operations of \L{}ukasiewicz logic. We will be dealing with five operations, namely negation ($\neg$), the strong and weak disjunction ($\oplus$ and $\lor$ respectively) and the strong and weak conjunction ($\odot$ and $\wedge$ respectively).

\begin{table}[h]
\centering
\small
\begin{tabular}{|l|l|} \hline
\textbf{Name} & \textbf{Definition} \\ \hline
Negation $\neg$ & $F_\neg(x)=1-\mu(x)$ \\ \hline
Strong disjunction $\oplus$ & $F_\oplus(x,y) = min\{1,x + y\}$ \\\hline 
Strong conjunction $\odot$ & $F_\odot(x,y) = max\{x + y-1,0\}$ \\\hline
Weak disjunction $\lor$ & $F_\lor(x,y) = max\{x,y\}$ \\\hline
Weak conjunction $\wedge$ & $F_\wedge(x,y)=min\{x,y\}$\\\hline
Implication $\rightarrow$ & $F_\rightarrow(x,y)=min\{1,1-x+y\}$\\
\hline
\end{tabular}
\caption{Logical operations in \L{}ukasiewicz logic}
\end{table}

Given a formula $\phi$ in \L{}ukasiewicz logic and an assignment $I$, we say that $I$ satisfies $\phi$ iff $[\phi]_I=1$.

\begin{example}
Let $\phi=\neg(x_1 \odot x_2 \odot \neg x_3)$. Consider the following two assignments:
\begin{enumerate}
\item $I_1$ with $I_1(x_1)=0, I_1(x_2)=0, I_1(x_3)=1$. $$[\phi]_{I_1}=\neg(max\{0+0-1,0\} \odot \neg 1) = \neg(0 \odot 0) = \neg 0 = 1$$
\item $I_2$ with $I_2(x_1)=0.6, I_2(x_2)=0.7, I_2(x_3)=0.2$. $$[\phi]_{I_2}=\neg(max\{.6+0.7 - 1,0\} \odot (\neg 0.2)) $$$$= \neg(max\{.3+0.8 - 1,0\}) = \neg 0.1 = 0.9$$
\end{enumerate}
So, $I_1$ satisfies $\phi$, but $I_2$ does not.
\end{example}

The same principle of satisfiability exists in fuzzy logics (and many-valued logics), denoted SAT$_\infty$. Like its classical counterpart, it is useful for solving a variety of problems. We say that a formula $\phi$ in \L{}ukasiewicz logic is \textit{satisfiable} iff there exists an assignment $I$ such that $[\phi]_I=1$.

\begin{example}
In the previous example, $\phi$ is satisfiable since there exists an assignment ($I_1$) that satisfies it.
\end{example}
An assignment $I$ is said to be a \textit{model} of a set of formulas $\Phi$ iff $l \leq [\alpha]_I \leq u$ for every formula $\alpha \in \Phi$, given a lower bound $l$ and upper bound $u$ for that formula (usually $u = 1$, and in classical logic even both $l= u=1$).

Solving satisfiability in fuzzy logics is still growing in theory as well as in application. In addition, to the best of our knowledge, MaxSAT has not been defined over fuzzy logic.

\subsubsection{Discretization}
In practice, it is common to assume a finite number of truth degrees, taken from a set $$T_k=\{0,\frac{1}{k},\frac{2}{k},\dots,1\}$$ with $k \in \mathbb{N}-\{0\}$.
\begin{itemize}
\item Let $L_\infty$ denote infinite-valued \L{}ukasiewicz logic and $L_k$ denote the $(k+1)$-valued version in which only interpretations are considered that take truth degrees from $T_k$.

\item For every set of formulas $\Phi$ in $L_\infty$, there exists a finite number of truth degrees $d$, such that $\Phi$ is satisfiable in $L_\infty$ iff it is satisfiable in $L_d$.
\end{itemize}

\section{Fuzzy MaxSAT}

\begin{definition}
Given a set of formulae $\Phi$ in \L{}ukasiewicz logic, the MaxSAT problem asks for an assignment $I$ that maximizes the number of satisfied formulae in $\Phi$.
\end{definition}

From now on we will call the MaxSAT problem defined over propositional logic ``Boolean MaxSAT'' and the one defined over \L{}ukasiewicz logic ``fuzzy MaxSAT''.

\begin{definition}
The fuzzy Partial MaxSAT problem (fuzzy PMaxSAT) for the \L{}ukasiewicz set of formulae $\phi=S \cup H$ is the problem of finding an assignment that satisfies all the formulae in $H$ and maximizes the number of satisfied formulae in $S$.
\end{definition}

\begin{example}
Consider solving fuzzy MaxSAT on $\phi$ in example 1. Assignment $I_1$ is a solution since $[\phi]_{I_1}=1$, which is the largest truth degree possible.
\end{example}

\begin{theorem}
Boolean MaxSAT is reducible to fuzzy PMaxSAT in polynomial time.
\end{theorem}
\begin{proof}
Let $\phi$ be a instance Boolean MaxSAT instance. We will construct a fuzzy MaxSAT instance $\phi'=S \cup H$ such that $I$ maximizes the number of satisfied clauses in $\phi$ iff $I$ maximizes the number of clauses in $\phi'$, where $I$ is an assignment.

We construct $S$ and $H$ as follows:
\begin{enumerate}
\item For every variable $x$ appearing in $H$, add the formula $\neg(x \oplus x)\oplus x$ to $\phi'$. This formula evaluates to 1 iff $x_k=0$ or $x_k=1$.
\item For every clause $(l_1 \lor \dots \lor l_i)$ in $\phi$, add the formula $(l_1 \oplus \dots \oplus l_i)$ to $S$.
\end{enumerate}

Thus $$H=\{\neg(x \oplus x)\oplus x \mid x \mbox{ appears in }\phi \}$$ and $$ S=\{(l_1 \oplus \dots \oplus l_i) \mid (l_1 \lor \dots \lor l_i) \in \phi \}$$

If the number of variables appearing in $\phi$ is $n$ and $\vert \phi \vert = m$, then $\vert \phi' \vert = n + m$ such that $\vert S \vert = m$ and $\vert H \vert = n$.

Assume that there are $k$ satisfied clauses in $\phi$. Then every variable $x$ evaluates to either 0 or 1. Thus, every $\neg (x \oplus x) \oplus x$ is satisfied and hence all $H$ is satisfied. If a clause $(l_1 \lor \dots \lor l_m)$ is satisfied, then $(l_1 \oplus \dots \oplus l_m)$ is also satisfied. Hence, there are exactly $k$ satisfied formulae in $S$.

Now assume that $I$ is a solution that satisfies $k$ clauses in $S$. Then surely every variable $x$ appearing in $\phi'$ has a value either 0 or 1. This is because $I$ certainly satisfies all formulae in $H$, which ensure just that. Since the semantics of the strong disjunction when restricted to 0 and 1 is identical to the semantics of Boolean disjunction, then if $(l_1 \oplus \dots \oplus l_i) \in S$ is satisfied then so is $(l_1 \lor \dots \lor l_i) \in \phi$. Therefore, $I$ satisfies exactly $k$ clauses in $\phi$.
\end{proof}

\section{Encodings}
Before presenting the encodings, it is important to note that one can generalize Boolean CNF by replacing the Boolean negation with the \L{}ukasiewicz negation and the Boolean disjunction with the strong disjunction. The resulting form is $$\bigwedge_{i=1}^m \left( \bigoplus_{j=1}^{r_i} l_{ij} \right)$$ and is referred to as \textit{simple \L{}-clausal form} in\cite{bofill2015finding}.

It has been shown\cite{bofill2015finding} that the satisfiability problem for any simple \L{}-clausal form is solvable in linear time, contrary to the SAT problem in Boolean logic which is NP-complete in the general case. In addition, the expressiveness of simple \L{}-clausal forms is limited. That is, not every \L{}ukasiewicz formula has an equivalent simple \L{}-clausal form. To remedy this matter, another form has been proposed called \textit{\L{}-clausal form}, for which the SAT problem is NP-complete\footnote{The proof involves reducing Boolean 3SAT to the SAT problem for \L{}-clausal forms.}.
\begin{definition}
Let $X=\{x_1,\dots,x_n \}$ be  a set of variables. A \textit{literal} is either a variable $x_i \in X$ or $\neg x_i$. A \textit{term} is a literal or an expression of the form $\neg(l_1 \oplus \dots \oplus l_k)$, where $l_1, \dots, l_k$ are literals. A \textit{\L{}-clause} is disjunction of terms. A \textit{\L{}-clausal form} is a conjunction of \L{}-clauses.
\end{definition}

\subsection{A Proposed Fuzzy MaxSAT Algorithm for Simple \L{}-clausal forms}
The proposed algorithm takes advantage of the fact that the SAT problem for simple \L{}-clausal forms is solvable in linear time. Moreover, it is based on the basic SAT-based technique of Boolean MaxSAT solving.

Let $\phi=\{C_1,\dots,C_m\}$ be a MaxSAT instance. The following formula is satisfiable iff there are $$\phi_k = \{C_1 \lor b_1,\dots,C_m \lor b_m\} \cup CNF(\sum_{i=1}^{m} b_i \leq k)$$   
where each $b_i,(1 \leq i \leq m)$ is a new variable and $CNF(\sum_{i=1}^{m} b_i \leq k)$ is the encoding of the constraint $\sum_{i=1}^{m} b_i \leq k$ to CNF. This constraint is satisfied if $\phi_k$ has at most $k$ falsified clauses. There are three ways to start searching for the value of $k$ which corresponds to the optimal solution, denoted $k_{opt}$:
\begin{enumerate}
\item Start at $k = 0$ and increase $k$ while $\phi_k$ is unsatisfiable.
\item Start at $k=m$ and decrease $k$ while $\phi_k$ is satisfiable.
\item Do binary search for $k_{opt}$: alternate between satisfiable $\phi_k$ and unsatisfiable $\phi_k$ until the algorithm converges to $k_{opt}$.

An interesting question is, can we use the same technique for simple \L{}-clausal forms? Remember that for such forms, the satisfiability problem is solvable in linear time, and thus the time complexity of the resulting algorithm will be a huge improvement over that of Boolean MaxSAT.
\end{enumerate}

\subsection{Into DLRs}
A \textit{disjunctive linear relation} (DLR) is an expression $$P_1 \circ_1 Q_1 \lor \dots \lor P_m \circ_m Q_m$$ where each $P_i$ and $Q_i$ is a polynomial of degree one with rational coefficients over the real-valued variables $X=\{x_1,\dots,x_n\}$ and each $\circ_i \in \{<,\leq,>,\geq,=,\neq \}$, ($1 \leq i \leq m$)\cite{fisher2005handbook}.

The satisfiability problem for a set $D$ of DLRs (denoted SAT$_{DLR}$) is determining whether there exists an assignment $I:X \rightarrow \mathbb{R}$ such that every DLR in $D$ is satisfied. In 1998, Jonsson and B{\"a}ckstr{\"o}m\cite{jonsson1998unifying} showed that SAT$_{DLR}$ is NP-complete.

Let $\Phi=\{\phi_1,\dots,\phi_m\}$ be a set of formulae. We replace each $\phi_i$ by $(\phi_i \rightarrow \neg y_i)$ where $y_i$ is a new variable, for $i \in \{1,\dots,m\}$. Each of these formulae ensure that if $\phi_i$ is satisfied then $y_i$ is falsified.

Each occurrence of $max\{a_1,\dots,a_k\}$ can be replaced by $(s)$ with the following inequalities: $(s \geq a_1),\dots,(s \geq a_k),(s=a_1 \lor \dots \lor s=a_k)$. The purpose of rewriting $min$ and $max$ is that they are nonlinear functions and they do not fit the formulation of DLRs. Each occurrence of $min\{a_1,\dots,a_n\}$ can be replaced and rewritten similarly. 

The final step is to add a bound on the $y_i,(1 \leq i \leq m)$ variables to capture the semantics of maximizing the number of satisfied formulae. Thus, we first add the bound $\sum_{i=1}^{m} y_i \leq 0$ and check the satisfiability of the DLR instance. If it is not satisfiable, then we keep increasing the bound until the instance is satisfiable.

\subsection{Into MILP}
A \textit{Mixed Integer Linear Program} (MILP) involves minimizing
	 $\sum_{j=1}^{n}c_jx_j$ subject to $$\sum_{j=1}^{n}a_{ij}x_j \leq b_i, (i=1,2,\dots,m)$$ $$x_j \geq 0, (j=1,2,\dots,n)$$ where some of $x_j,(1 \leq j \leq n)$ are integers and some are real and $c_j, a_{ij},b_i \in \mathbb{R}$
	 
Let $\Phi$ be a set of $r$ formulas $\{\phi_1,\dots,\phi_m\}$.
\begin{itemize}
\item Introduce $m$ new binary variables $y_1,\dots,y_m$ (one per formula).
\item Replace each formula $\phi_i$ by the relaxed formula $\neg \phi_i \rightarrow y_i$.
\item Take the resulting formula of the previous step and replace every occurrence of $min$ and $max$ due to $\neg, \oplus, \odot,\lor$ and $\wedge$ just as shown in the DLR encoding.
\end{itemize}

Let $f(\phi_i)$ be the result of performing the previous three steps on $\phi_i$. The MILP instance would be:
$$\mbox{Minimize }\sum_{i=1}^{m}y_i \mbox{ subject to } f(\phi_1) \geq 1, \dots, f(\phi_m) \geq 1$$

\subsection{Into WCSP}
A weighted CSP (WCSP) instance is a triple $(X,D,C)$, where $X=\{x_1,\dots,x_n\}$, $D=\{d(x_1).\dots,d(x_n) \}$ and $C=\{C_1,\dots,C_r \}$ are variables, domains and constraints respectively. Each $C_i \in C$ is a pair $(S_i,f_i)$, where $S_i =\{x_{i_1},\dots,x_{i_k} \}$ is the constraint scope and $f_i : d(x_{i_1}) \times \dots \times d(x_{i_k}) \rightarrow \mathbb{N}$ is a cost (weight) function that maps each tuple to its associated weight. 
 
An optimal solution to a WCSP instance is a complete assignment to the variables in $X$ in which the sum of the costs of the constraints is minimal. The WCSP Problem for a WCSP instance consists in finding an optimal solution for that instance.

Given a set of formulae $\Phi=\{\phi_1,\dots,\phi_m\}$, we encode the problem as follows:
\begin{enumerate}
\item Create a variable $x_i$ for each formula $\phi_i$ with domain $d(x_i)$ the set of possible assignments to the variables appearing in $\phi_i$. When $x_i$ takes a value in $d(x_i)$ this represents the fact that the variables of $\phi_i$ have been assigned accordingly.

\item For each variable $x_i$ add a constraint that assigns cost 0 to each domain value satisfying $\phi_i$ and assigns cost $\infty$  to the values falsifying $\phi_i$.

\item For each two formulae $\phi_i$ and $\phi_j$ sharing variables, we add a constraint that assigns cost $\infty$ to assignments that assign different values to the shared variables, and cost 0 otherwise.
\end{enumerate}

\section{Future Work}
Recently, it has been found that the satisfiability problem for simple \L{}-clausal forms can be solved in linear time. We will investigate whether or not an algorithm for the new MaxSAT problem for simple \L{}-clausal forms can take advantage of this fact in order for its time complexity to be polynomial. 

An interesting alternative definition to fuzzy MaxSAT is: Given a formula $\phi$ in \L{}ukasiewicz logic, find an assignment $I$ such that $[\phi]_I$ is maximum. In other words, the new definition asks for an assignment that maximizes $\phi$'s truth degree.

\bibliographystyle{plain}
\bibliography{references}
\end{document}